\documentclass{llncs}

\RequirePackage{amsmath}

\usepackage{amssymb}
\usepackage{xspace}
\usepackage{color}
\usepackage{booktabs}
\usepackage[utf8]{inputenc}
\usepackage{comment}
\usepackage{pbox}
\usepackage{url}
\usepackage{graphicx}
\usepackage[shortlabels]{enumitem}

\usepackage{cite}

\usepackage[hidelinks]{hyperref}
\hypersetup
{
    pdfauthor={Aleksander Cis{\l}ak, Szymon Grabowski},
    pdftitle={Lightweight Fingerprints for Fast Approximate Keyword Matching Using Bitwise Operations},
    pdfkeywords={fingerprint} {approximate} {string matching} {bitwise},
    urlcolor=black,
    colorlinks=false
}

\begin{document}

\title{Lightweight Fingerprints for Fast Approximate Keyword Matching Using Bitwise Operations}
\author{\texorpdfstring{Aleksander~Cis{\l}ak$^1$~and~Szymon Grabowski$^2$}{Aleksander~Cis{\l}ak~and~Szymon~Grabowski}}

\institute{$^1$\,Warsaw University of Technology, Faculty of Mathematics and Information Science,\\
ul.~Koszykowa 75, 00--662 Warsaw, Poland\\
\email{a.cislak@mini.pw.edu.pl}\\
$^2$\,Lodz University of Technology, Institute of Applied Computer Science,\\
  Al.~Politechniki 11, 90--924 {\L}\'od\'z, Poland\\
  \email{sgrabow@kis.p.lodz.pl}\\
}

\maketitle

\begin{abstract}
We aim to speed up approximate keyword matching by storing a lightweight, fixed-size block of data for each string, called a fingerprint.
These work in a similar way to hash values; however, they can be also used for matching with errors.
They store information regarding symbol occurrences using individual bits, and they can be compared against each other with a constant number of bitwise operations.
In this way, certain strings can be deduced to be at least within the distance $k$ from each other (using Hamming or Levenshtein distance) without performing an explicit verification.
We show experimentally that for a preprocessed collection of strings, fingerprints can provide substantial speedups for $k = 1$, namely over $2.5$ times for the Hamming distance and over $10$ times for the Levenshtein distance.
Tests were conducted on synthetic and real-world English and URL data.
\end{abstract}

\section{Introduction}

This study is concerned with strings, that is, finite collections of symbols.
We assume that a string $S$ is $1$-indexed, i.e.~index $1$ refers to the first symbol $S[1]$, index $2$ refers to the second symbol $S[2]$, etc.
All strings are specified over the same alphabet $\Sigma$, with alphabet size $\sigma = |\Sigma|$.

Exact string comparison refers to checking whether two strings $S_1$ and $S_2$ of equal length $n$ have the same characters at all corresponding positions. 
The strings can store, e.g.,~natural language data or DNA sequences.
Assuming that each character occupies $1$ byte, calculation of such a comparison takes $O(n)$ time in the worst case; however, the average case is $O(1)$, using no additional memory.
Specifically, the complexity of the average case of comparing two strings depends on the alphabet size.
Assuming a uniformly random symbol distribution, the chance that first two symbols match (i.e.~that $S_1[1] = S_2[1]$) is equal to $1/\sigma$, the chance that both first and second symbol pairs match (i.e.~that $S_1[1] = S_2[1]$ and $S_1[2] = S_2[2]$) is equal to $1/\sigma^2$, etc.
More generally, the probability that there is a match between all characters up to a $1$-indexed position $i$ is equal to $1 / \sigma^i$.

Nonetheless, even when one can expect a fair degree of similarity between the strings, it is often faster to compare hash values for two strings (in constant time) and perform an explicit verification only when these hashes are equal to each other.
This is particularly true in a situation where one would compare a single string, that is a query pattern, against a preprocessed collection (dictionary) of strings.
The hash-based approach forms the basis of, e.g.,~the well-known \mbox{Rabin--Karp}~\cite{karp1987efficient} algorithm for online exact matching.

Aside from exact matching, there has been a substantial interest in \emph{approximate} string comparison, for instance for spelling suggestions or matching biological data~\cite{navarro2001guided, pollock1984automatic, myers1994sublinear}.
Approximate string matching defines whether two strings are equal according to a specified similarity metric, and the number of errors is denoted by $k$.
Two popular measures include $(i)$ the Hamming distance~\cite{hamming1950error} (later referred to as $Ham$), which defines the number of mismatching characters at corresponding positions between two strings of equal length, and $(ii)$ the Levenshtein distance~\cite{levenshtein1966binary} (also called edit distance, later referred to as $Lev$), which determines the minimum number of edits (insertions, deletions, and substitutions) required for transforming one string into another.

Hash values cannot be easily used in the approximate context.
This work is focused on approximate matching in practice, and hence we introduce the concept of lightweight \emph{fingerprints}, whose goal is to speed up approximate string comparison.
The speedup can be achieved for preprocessed collections of strings, at the cost of a fixed-sized amount of space per each word in the collection.
This means that we evaluate fingerprints for a keyword indexing problem, also known as dictionary (keyword) matching.
Specifically, in this setting a pattern $P$ is compared against a string collection $\mathcal{D} = [S_1, \dots, S_{|\mathcal{D}|}]$.

\section{Related work}

The original idea of a string fingerprint (which is referred to as a sketch in some articles) goes back to the work of Rabin and Karp~\cite{rabin1981fingerprinting, karp1987efficient}.
They used a variant of a hash function called a rolling hash, which can be quickly calculated for each successive substring of the input text, in order to speed up exact online substring matching.
This technique was later used also in the context of multiple pattern matching~\cite{MuthM96,SalmelaTK06} and matching over a 2D text~\cite{ZhuT89}.
Bille et al.~\cite{bille2013fingerprints} extended this idea and demonstrated how to construct fingerprints for substrings of a string which is compressed by a context-free grammar.
Policriti et al.~\cite{PolicritiTV12} generalized the classical Rabin--Karp algorithm in order to be used with the Hamming distance.

Fingerprints can be viewed as a method of compressing the text, nevertheless, they cannot replace the text -- rather, they can be used as additional information.
Bar-Yossef~et~al.~\cite{bar2004sketching}~show that it is not possible to use only a fingerprint in order to answer a matching query. 
Moreover, they prove that for answering Hamming distance queries, the size of the fingerprint is bounded by $\Omega(n/m)$, where $n$ is the text length and $m$ is the pattern length.

Prezza and Policriti~\cite{PolicritiP14} presented a related idea called de Bruijn hash function, where shifting the substring by one character results in a corresponding one-character shift in its hash value.
Grabowski and Raniszewski~\cite{GrabowskiR15} used fingerprints in order to speed up verifications of tentative matches in their SamSAMi (sampled suffix array with minimizers) full-text index.
Fingerprints, which are concatenations of selected bits taken from a short string, allow them to reject most candidate matches without accessing the indexed text and thus avoiding many cache misses.
Recently, fingerprints have been applied to the longest common extension (LCE) problem~\cite{prezza2016place}, allowing to solve the LCE queries in logarithmic time in essentially the same space as the input text (replacing the text with a data structure of the same size).

Ramaswamy et al.~\cite{ramaswamy2006approximate} described a technique called ``approximate'' fingerprinting; however, it refers to exact pattern matching with false positives rather than matching based on similarity metrics.
Fingerprints have also been used for matching at a larger scale, i.a.,~for determining similarity between audio recordings~\cite{CanoBKH05} and files~\cite{Manber94}.
The term fingerprint has also been used with a different meaning in the domain of string processing, where it refers to the set of distinct characters contained in one of the substrings of a given string, with the ongoing recent work, e.g.,~a study by Belazzougui et al.~\cite{BelazzouguiKR13}.

\section{Fingerprints}
\label{Sec:our_alg}

In this section we introduce the notion of a fingerprint, describe its construction and demonstrate how to compare two fingerprints.
For a given string $S$, a fingerprint $S^\prime$ is constructed as $S^\prime = f(S)$ using a function $f$ which returns a fixed-sized block of data.
In particular, for two strings $S_1$ and $S_2$, we would like to determine in certain cases that $Ham(S_1, S_2) \geqslant k$ or $Lev(S_1, S_2) \geqslant k$  by comparing only fingerprints $S^\prime_1$ and $S^\prime_2$ for a given $k \in \mathbb{N}_{>0}$.
In other words, fingerprints allow for quickly rejecting an approximate match up to $k$ errors between two strings.
When the fingerprint comparison is not decisive, we still have to perform an explicit verification on $S_1$ and $S_2$, but fingerprints allow for reducing the overall number of such operations.
There exists a similarity between fingerprints and hash functions; however, hash comparison works only in the context of exact matching.

As regards the complexity of a single verification (string comparison), the worst case is equal to $O(n)$ for the Hamming distance and $O(k \,\text{min}(|S_1|, |S_2|))$ for the Levenshtein distance, using Ukkonen's algorithm~\cite{ukkonen1985algorithms}.
Assuming a uniformly random alphabet distribution, the average case complexity is equal to $O(k)$ for both metrics.

In our proposal, fingerprints use individual bits in order to store information about symbol frequencies or positions in the string $S[1, n]$.
Let $\Sigma' \subseteq \Sigma$ be a subset of the original alphabet with $\sigma' = |\Sigma'|$ denoting its size.
We propose the following approaches:

\begin{itemize}
\item
\textbf{Occurrence} (\textbf{occ} in short): we store information in each bit that indicates whether a certain symbol from $\Sigma'_{occ}$ occurs in a string using $\sigma'_{occ}$ bits in total.

\item
\textbf{Occurrence halved}: the fingerprint refers to occurrences in the first and second halves of $S$, that is $S[1, \lfloor n/2 \rfloor]$ and $S[\lfloor n/2 \rfloor + 1, n]$, respectively.
We store information whether each of the $\sigma'_{occh}$ symbols occurs in the first half of $S$ using the first $\sigma'_{occh}$ bits of the fingerprint, and we store information whether each of the same $\sigma'_{occh}$ symbols occurs in the second half of $S$ using the second $\sigma'_{occh}$ bits of the fingerprint.
The occurrence halved scheme works only for the Hamming distance.

\item
\textbf{Count}: we store a count (i.e.~the number of occurrences) of each symbol using $b$ bits per symbol.
The count can be in the range $[0, 2^b - 1]$, where $2^b - 1$ indicates that there are $2^b - 1$ or more occurrences of a given symbol.
We use $\sigma'_{count}$ symbols from $\Sigma'_{count}$.

\item
\textbf{Position} (\textbf{pos} in short): we can encode information regarding the first (leftmost, i.e.~the one with the lowest index) position in $S$ of each symbol from $\Sigma'_{pos}$ using $p$ bits per symbol, where $p \leqslant \lceil \text{log}_2 n \rceil$.
This position can be in the range $[1, 2^p - 1]$ (encoded in the fingerprint as 0-indexed, i.e.~where index 0 refers to the first symbol, index 1 refers to the second symbol, etc), and the value of $2^p - 1$ indicates that the first occurrence is either at one of the positions from the range $[2^p, n]$ or the symbol does not occur in $S$ (we do not know which one is true).
We use $\sigma'_{pos} \cdot p$ bits in order to encode positions of $\sigma'_{pos}$ symbols.
The remaining bits, e.g.,~$1$ bit for $\sigma'_{pos} = 5$, $p = 3$ and $16$ bits per fingerprint, are used in order to store information about the occurrences of additional symbols, in the same fashion as in the occurrence fingerprint which was introduced previously.
The position-based scheme works only for the Hamming distance.

\end{itemize}

Fingerprints can be also differentiated based on the symbols which they refer to.
The choice of the specific symbol set is important when it comes to an empirical evaluation and it is discussed in more detail in Section~\ref{Sec:results}.
We have identified the following possibilities:

\begin{itemize}

\item
\textbf{Common}: a set of symbols which appear most commonly in a given collection.

\item
\textbf{Rare}: a set of symbols which appear least commonly in a given collection.

\item
\textbf{Mixed}: a mixed set where half of the symbols comes from the common set while the other half comes from the rare set.

\end{itemize}

\subsection{Fingerprint examples}

In the following examples, we constrain ourselves to the variant of 2-byte (16-bit) fingerprints with common letters.
Fingerprints could in principle have any size, and the longer the fingerprint, the more information we can store about the character distribution in the string.
Still, we regard 2 bytes, which correspond to the size of 2 characters in the original string, to be a desirable compromise between size and performance (consult the following section for experimental results).
The choice of common letters is arbitrary at this point and it only serves the purpose of idea illustration.

In the following examples, occurrence fingerprint is constructed using selected 16 most common letters of the English alphabet, namely $\{ \texttt{e}, \texttt{t}, \texttt{a}, \texttt{o}, \texttt{i}, \texttt{n}, \texttt{s}, \texttt{h}, \texttt{r}, \texttt{d},\\ \texttt{l}, \texttt{c}, \texttt{u}, \texttt{m}, \texttt{w}, \texttt{f} \}$ \cite[p.~36]{lewand2000cryptological}.
For the occurrence halved and count fingerprints (with $b = 2$ bits per count), we use the first 8 letters from this set.
In the case of a position fingerprint (with $p = 3$ bits per letter), we use the first 5 letters for storing their positions and the sixth letter \texttt{n} for the last (single) occurrence bit.

Each fingerprint type would be as follows for the word \texttt{instance} (spaces are added only for visual presentation):

\begin{itemize}
\item
\textbf{Occurrence}:

\vspace{0.2em}
\boxed{\texttt{1\,1\,1\,0\,1\,1\,1\,0\,0\,0\,0\,1\,0\,0\,0\,0}}
\vspace{0.2em}

The first (leftmost) bit corresponds to the occurrence of the letter \texttt{e} (which does occur in the word, hence it is set to \texttt{1}), the second bit corresponds to the occurrence of the letter \texttt{t}, etc.

\item
\textbf{Occurrence halved}:

\vspace{0.2em}
\boxed{\texttt{01~10~01~00~10~11~10~00}}
\vspace{0.2em}

The first (leftmost) bit corresponds to the occurrence of the letter \texttt{e} in the first half of the word, that is \texttt{inst}; the second bit corresponds to the occurrence of the letter \texttt{e} in the second half of the word, that is \texttt{ance}; the third bit corresponds to the occurrence of the letter \texttt{t} in the first half of the word, the fourth bit corresponds to the occurrence of the letter \texttt{t} in the second half of the word, etc.

\item
\textbf{Count}:

\vspace{0.2em}
\boxed{\texttt{01~01~01~00~01~10~01~00}}
\vspace{0.2em}

The first two (leftmost) bits correspond to the count of the letter \texttt{e} (it occurs once, hence the count is \texttt{01}, that is $1$), the second two bits correspond to the count of the letter \texttt{t} (it occurs once, hence the count is \texttt{01}, that is $1$), etc.

\item
\textbf{Position}:

\vspace{0.2em}
\boxed{\texttt{111 011 100 111 000 1}}
\vspace{0.2em}

The first three (leftmost) bits correspond to the position of the first occurrence of the letter \texttt{e} (this $0$-indexed position is equal to $7$, hence it is set to \text{111}), the second three bits correspond to the position of the first occurrence of the letter \texttt{t} (this $0$-indexed position is equal to $3$, hence it is set to \texttt{011}), etc.
The last (rightmost) occurrence bit indicates the occurrence of \texttt{n}, and since this letter does occur in the input string, this bit is set to \texttt{1}.

\end{itemize}

\subsection{Construction}

The construction of various fingerprint types is described below.
For the description of symbols and types, consult preceding subsections.
At the beginning, all bits for the fingerprint are always set to \texttt{0}.

\begin{itemize}
\item
\textbf{Occurrence}:
Let us remind the reader that the length of the fingerprint is equal to $\sigma_{occ}'$ for a selected alphabet $\Sigma_{occ}'$ of letters whose occurrences are stored.
A string is iterated characterwise.
For each character $c$, a mask \texttt{0x1} is shifted $q$ times to the left, where $q \in \{0, \dots, \sigma_{occ}' - 1 \}$ is a corresponding shift for the character $c$.
In other words, there exists a mapping $c \to q$ for each character $c \in \Sigma_{occ}'$.
A natural approach to this mapping is to take the position of a symbol in the alphabet $\Sigma_{occ}'$ (assuming that the alphabet is ordered).
The fingerprint is then \texttt{or}-ed with the mask in order to set the bit which corresponds to character $c$ to \texttt{1}.
For a string of length $n$, time complexity of this operation is equal to $O(n)$.

\item
\textbf{Occurrence halved}:
The fingerprint is constructed in an analogous way to the occurrence approach described above.
We start with iterating the first half of the string, setting corresponding bits depending on letter occurrences, and then we iterate the second half of the string, again setting corresponding bits, which are shifted by $1$ position with respect to bits set while iterating the first half of the string.
Character mapping is adapted accordingly.

\item
\textbf{Count}:
A string is again iterated characterwise.
The length of the fingerprint is equal to $b \cdot \sigma_{count}'$ for a selected alphabet $\Sigma_{count}'$ of letters whose counts are stored.
Similarly to the occurrence fingerprint, there exists a mapping $c \to q$ for each character $c \in \Sigma_{count}'$.
However, since we need $b$ bits in order to store a count, it holds that $q \in \{0, b, \dots, \sigma_{count}' - b \}$, assuming that $b$ divides $\sigma_{count}'$.
A selected bit mask is set and the fingerprint is then \texttt{or}-ed with the mask in order to increase the current count of character $c$ which is stored using $b$ bits at positions $\{q, q + 1, \dots, q + b - 1\}$.
For instance, for a practical case of $b = 2$, the count can be in the range $\{0, 1, 2, 3\}$.
In this case, we would proceed as follows: if the right bit is \texttt{0}, then we set it to \texttt{1} (either increasing the count from $0$ to $1$ or from $2$ to $3$).
Otherwise, if the left bit is set, we also set the right bit (increasing the count from $2$ to $3$), or otherwise we set the left bit and unset the right bit (increasing the count from $1$ to $2$).
All of the above steps can be realized with a few simple bitwise operations.
Assuming fixed $b$, for a string of length $n$, the time complexity of this operation is equal to $O(n)$.

\item
\textbf{Position}:
In the case of position fingerprints, the length of the fingerprint is equal to $\sigma_{pos}' \cdot p$ for a selected alphabet $\Sigma_{pos}'$ of letters whose positions are stored and a chosen constant $p$ which indicates the number of bits per position.
Here, we iterate the alphabet, and for each character $c \in \Sigma_{pos}'$ we search for the first (leftmost) occurrence of $c$ in the string.
Each position of such an occurrence is then successively encoded in the fingerprint, or the position $pos$ is set to all \texttt{1}s if $pos \geqslant 2^p - 1$.
For a string of length $n$, the time complexity of this operation is equal to $O(n \cdot \sigma_{pos}')$.

\end{itemize}

\subsection{Comparison}

We can quickly compare two \emph{occurrence} (or occurrence halved) fingerprints by performing a binary \texttt{xor} operation and counting the number of bits which are set in the result (that is, calculating the Hamming weight, $H_W$).
Let us note that $H_W$ can be determined in constant time using a lookup table with $2^{8|S^\prime|}$ entries, where $|S^\prime|$ is the fingerprint size in bytes.
We denote the fingerprint distance with $F_D$, and for occurrence fingerprints $F_D(S^\prime_1, S^\prime_2) = H_W(S^\prime_1 \oplus S^\prime_2)$.
In other words, we count the number of mismatching character occurrences which are stored in individual bits.

Let us however note that $F_D$ does not determine the true number of errors.
For instance for $S_1 = \texttt{run}$ and $S_2 = \texttt{ran}$, $F_D$ might be equal to 2 (occurrence differences for \texttt{a} and \texttt{u}) but there is still only one mismatch.
On the other extreme, for two strings of length $n$, where each string consists of a repeated occurrence of one different symbol, $F_D$ might be equal to 1 (or even 0, if the symbols are not included in the fingerprints), but the number of mismatches is $n$.
In general, $F_D$ can be used in order to provide a lower bound on the true number of errors, and the following relation holds (the right-hand side can be calculated quickly using a lookup table, since $0 \leqslant F_D \leqslant 8 |S^\prime|$):

\begin{equation} \label{eq:fingerprint}
D(S_1, S_2) \geqslant \lceil F_D(S^\prime_1, S^\prime_2) / 2 \rceil, D \in \{Ham, Lev\}
\end{equation}

This formula also holds for the \emph{count} fingerprint, although the true number of errors is underestimated even more in such a case, since we calculate the Hamming weight instead of comparing the counts.
For instance, for the count of $3$ (bits \texttt{11}) and the count of $1$ (bits \texttt{01}), the resulting difference is equal to $1$ instead of $2$.

As regards the \emph{position} fingerprint (which is relevant only to the Hamming metric), after calculating the \texttt{xor} value, we do not compute the Hamming weight, rather, we compare each set of bits ($p$-gram) which describes a single position.
The value of $F_D$ is equal to the number of mismatching $p$-grams.
Similarly to other fingerprint types, these values can be preprocessed and stored in a lookup table in order to reduce calculation time.

The relationship between the fingerprint error and the true number of errors is further explored in Theorem~\ref{th:fingerprint1} and Theorem~\ref{th:fingerprint2}.
In plain words, manipulating a single symbol in either string makes the fingerprint distance grow by at most $2$.
Let us note that Formula~\ref{eq:fingerprint} follows as a direct consequence of this statement, with the round-up on the right-hand side resulting from the fact that fingerprint distance might be odd.

\begin{theorem}
\label{th:fingerprint1}

Consider $\mathcal{F} = \{occ, count \}$ and assume a distance function $D \in \{Ham, Lev\}$.
For any two strings $S_1$ and $S_2$, with their fingerprints $S^\prime_1$ and $S^\prime_2$, respectively, and the fingerprint distance between them $F_D(S^\prime_1, S^\prime_2) = f(S^\prime_1, S^\prime_2)$, where $f \in \mathcal{F}$, we have that for any string $S_3$ such that $D(S_2, S_3) = 1$, the following relation holds: $F_D(S^\prime_1, S^\prime_3) \leqslant F_D(S^\prime_1, S^\prime_2) + 2$.

\end{theorem}

\begin{proof}
\label{pr:fingerprint1}

Let us first consider the occurrence fingerprints and Hamming distance (i.e.~$D = Ham$).
For this distance, two strings must be of equal length (otherwise the distance is infinite) and let us set $|S_1| = |S_2| = n$.
The string $S_2$ can be obtained from $S_1$ by changing some of its
$k = D(S_1, S_2)$ symbols, at positions $1 \leqslant i_1 < i_2 < \ldots < i_k \leqslant n$.
Let $V_0$ be an initial copy of $S_1$ and in $k$ successive steps we transform it
into $V_1, V_2, \ldots, V_k = S_2$, by changing one of its symbols at a time.
For clarity, we shall modify the symbols in the order of their occurrence in the strings (from left to right).
We will observe how the changes affect the value of $F_D(S^\prime_1, V_j^\prime)$,
which is initially (i.e.~for $j = 0$) equal to zero.

Consider a $j$-th step, for any $1 \leqslant j \leqslant k$.
We have four cases:
\begin{enumerate}[$(i)$]
\item
both $V_{j-1}[i_j] \in \Sigma'$ and $V_j[i_j]\in \Sigma'$,
\item
both $V_{j-1}[i_j] \not \in \Sigma'$ and $V_j[i_j] \not \in \Sigma'$,
\item
$V_{j-1}[i_j] \in \Sigma'$ but $V_j[i_j] \not \in \Sigma'$,
\item
$V_{j-1}[i_j] \not \in \Sigma'$ but $V_j[i_j] \in \Sigma'$.
\end{enumerate}
Let us notice that:\\
in case $(i)$ $H_W(V_j^\prime) - H_W(V_{j-1}^\prime) \in \{-1, 0, 1\}$, yet
since $V_{j-1}[i_j] \ne V_j[i_j]$, we may obtain new mismatches at (at most) two positions
of the fingerprints, i.e.~$F_D(S^\prime_1, V_j^\prime) - F_D(S^\prime_1, V_{j-1}^\prime) \leqslant 2$,\\
in case $(ii)$ $V_j^\prime = V_{j-1}^\prime$ and thus $F_D(S^\prime_1, V_j^\prime) = F_D(S^\prime_1, V_{j-1}^\prime)$,\\
in case $(iii)$ $H_W(V_{j-1}^\prime) - H_W(V_j^\prime) \in \{0, 1\}$ and $F_D(S^\prime_1, V_j^\prime) - F_D(S^\prime_1, V_{j-1}^\prime) \leqslant 1$,\\
in case $(iv)$
$H_W(V_j^\prime) - H_W(V_{j-1}^\prime) \in \{0, 1\}$
and $F_D(S^\prime_1, V_j^\prime) - F_D(S^\prime_1, V_{j-1}^\prime) \leqslant 1$.\\

From the shown cases and by the triangle inequality we conclude that
replacing a symbol with another makes the fingerprint distance grow by at most 2.

Now we change the distance measure to the Levenshtein metric (that is, we set $D = Lev$).
Note that the set of available operations transforming one string into another is extended;
not only substitutions are allowed, but also insertions and deletions.
The overall reasoning follows the case of Hamming distance, yet we need to consider all three operations.
A single substitution in $V_j$, for a $j$-th step,
makes the fingerprint distance grow by at most 2, in the same manner as shown above for the Hamming distance.
Inserting a symbol $c$ into $V_j$ (at any position) implies one of three following cases:\\
$(i)$ $c \not \in \Sigma'$, where the fingerprint distance remains unchanged,\\
$(ii)$ $c \in \Sigma'$ and $c \in S_1$, where again the fingerprint distance does not change, or\\
$(iii)$ $c \in \Sigma'$ and $c \not \in S_1$, where the fingerprint distance grows by 1.

Deleting a symbol $c$ from $V_j$ (at any position) implies one of three following cases:\\
$(i)$ $c \not \in \Sigma'$, where the fingerprint distance remains unchanged (same as for the insert operation),\\
$(ii)$ $c \in \Sigma'$ and $c \in S_1$, where the fingerprint distance might not change
(if $V_j$ contains at least two copies of $c$) or it might grow by 1, or\\
$(iii)$ $c \in \Sigma'$ and $c \not \in S_1$, where the fingerprint distance might not change or it might decrease by 1.
Note, however, that the last case for the delete operation never occurs in an edit script transforming $S_1$ into $S_2$ using a minimum number of Levenshtein operations.

Handling $f = count$ is analogous to the presented reasoning for $f = occ$, both for the Hamming and the Levenshtein distance.
\qed
\end{proof}

\begin{theorem}
\label{th:fingerprint2}

Consider $F_D = pos$ and assume a distance function $D = Ham$.
For any two strings $S_1$ and $S_2$, with their fingerprints $S^\prime_1$ and $S^\prime_2$, respectively, we have that for any string $S_3$ such that $D(S_2, S_3) = 1$, the following relation holds: $F_D(S^\prime_1, S^\prime_3) \leqslant F_D(S^\prime_1, S^\prime_2) + 2$.
\end{theorem}

\begin{proof}
\label{pr:fingerprint2}
For the Hamming distance, two strings must be of equal length and let us set $|S_1| = |S_2| = n$.
Similarly to the case of occurrence and count fingerprints, the string $S_2$ can be obtained from $S_1$ by changing some of its $k = D(S_1, S_2)$ symbols.
This proof follows the same logic as presented in proof for Theorem~\ref{th:fingerprint1}.
Let us note that the only difference lies in the fact that we compare the first position of a given letter rather than its occurrence.
We deal with the same four cases depending on whether a modified letter belongs to $\Sigma^\prime$, and modifying a single letter may change: in case $(i)$ at most two $p$-grams (which describe the position of the first occurrence of a given letter), in case $(ii)$ 0 $p$-grams, in case $(iii)$ at most one $p$-gram, and in case $(iv)$ at most one $p$-gram (that is, the result of these two latter cases is equivalent).
Again, as before, the number of modified $p$-grams corresponds directly to the maximum change in fingerprint distance.
\qed
\end{proof}

\subsection{Storage}

Even though the true distance is higher than the fingerprint distance $F_D$, fingerprints can still be used in order to speed up comparisons because certain strings will be compared (and rejected) in constant time using only a fixed number of fast bitwise operations and array lookups.
As mentioned before, we consider a scenario where a number of strings is preprocessed and stored in a collection.
Since the construction of a fingerprint for the query string might be time-consuming, fingerprints are useful when the number of strings in a collection is relatively high.
As regards the space overhead incurred by the fingerprints, for a dictionary $\mathcal{D}$ containing $|\mathcal{D}|$ words it is equal to $O(|\mathcal{D}| F + \sigma)$, $F$ being the fingerprint size in bytes.
This holds since we have to store one constant size fingerprint per word together with the lookup tables which are used in order to speed up fingerprint comparison.
Let us note that this overhead is relatively small, especially when the size of each string is high (this is further discussed in the next section).

\section{Empirical study}
\label{Sec:results}

Experimental results were obtained on the machine equipped with the Intel i7-6700K processor running at $4.2$\,GHz and DDR4 memory ($3.0$\,GHz, 15-15-15-36 latency).
Source code of the tested implementation is available upon request.

The following data sets were used in order to obtain the experimental results:

\begin{itemize}
\item
\textbf{synthetic-eng} (synthetic English data): 10.0\,MB, generated based on English language letter frequencies~\cite[p.~36]{lewand2000cryptological}
\item
\textbf{iamerican-insane} (real-world English data): 5.89\,MB, American English language dictionary, available from Linux packages
\item
\textbf{urls} (real-world URL data): 90.62\,MB of web addresses, available online: \url{http://data.law.di.unimi.it/webdata/in-2004/}
\end{itemize}

The number of queries was equal to $1\,000$, and each evaluation was run $100$ times and an arithmetic mean was calculated.
All data refers to single-thread performance.
For the calculation of the Hamming distance, a regular loop which compares each consecutive character until $k$ mismatches are found was used.
It turned out that this implementation was faster than any other low-level approach (e.g.,~directly using certain processor instructions from the SSE extension set) when full compiler optimization was used.
For the Levenshtein distance, we used our own implementation based on the optimal calculation of the $2k+1$ strip and the 2-row window~\cite{gusfield1997algorithms}.
It turned out to be faster than publicly available implementations, for instance the version from the Edlib library~\cite{vsovsic2017edlib} or the SeqAn library~\cite{DoringWRR08}.
This was probably caused by the fact that we could use the most lightweight solution and thus omit certain layers of abstractions from the libraries, especially since the comparison function was invoked multiple times for relatively short strings.

Queries were extracted randomly from the dictionary and compared against this dictionary.
We have also tried distorting the queries by inserting a number of errors, each with a $50\,\%$ probability, and the results were consistent for any maximum number of introduced errors.
In the case of the English language text, we have also tested queries which consisted of the most common words extracted from a large corpus of the English language, and identical behavior was observed as in the case of queries which were sampled from the dictionary.
This test was performed in order to check whether words which would be more likely to be searched for in practice exhibit the same behavior.

Each fingerprint occupied $2$ bytes, since $1$-byte fingerprints turned out to be ineffective, and we regarded this as the optimal value with respect to a reasonable word size.
The mode length in the English dictionary was equal to $9$, which means that each fingerprint roughly incurred a $22\,\%$ storage penalty on average; however, the mode length in the URL collection was equal to $69$, which means that each fingerprint roughly incurred only a $3\,\%$ storage penalty on average.
Count fingerprints used $2$ bits per count, that is we set $b = 2$ and position fingerprints used $3$ bits per position, that is we set $p = 3$ (consult Section~\ref{Sec:our_alg} for details).
Given the selected fingerprint size of $2$ bytes ($16$ bits), these values allow for the use of $8$ letters for count fingerprints and $5$ letters for position fingerprints, with an extra occurrence bit in the latter case.

In our implementation, fingerprint comparison requires performing one bitwise operation and one array lookup, i.e.~2 constant operations in total.
We analyze the comparison time between two strings using various fingerprint types versus an explicit verification.
When the fingerprint comparison was not decisive, a verification consisting in distance calculation was performed and it contributed to the elapsed time.
The fingerprint is calculated once per query and it is then reused for the comparison with consecutive words, i.e.~we examine the situation where a single query is compared against a set (dictionary) of words.

\subsection{Results}

Figure~\ref{Fig:word_size_synth} demonstrates the results for synthetic English data, which allowed us to check a wide range of word sizes (which occur infrequently or not at all in natural language corpora) for occurrence, count, and position fingerprints.
Hamming distance was used as a similarity metric in this case.
As described in the previous section, common, mixed, and rare letter sets were selected based on English alphabet letter frequencies.
We can observe that the effectiveness of various approaches depends substantially on the word size, and the performance of letter sets also depends on the fingerprint type.
The highest speedup was provided by count fingerprints for common letters in the case of words of 18 characters and it was equal to roughly $2.5$ with respect to the naive comparison.

\begin{figure}[ht!]
    \centering
    \includegraphics[scale=0.43]{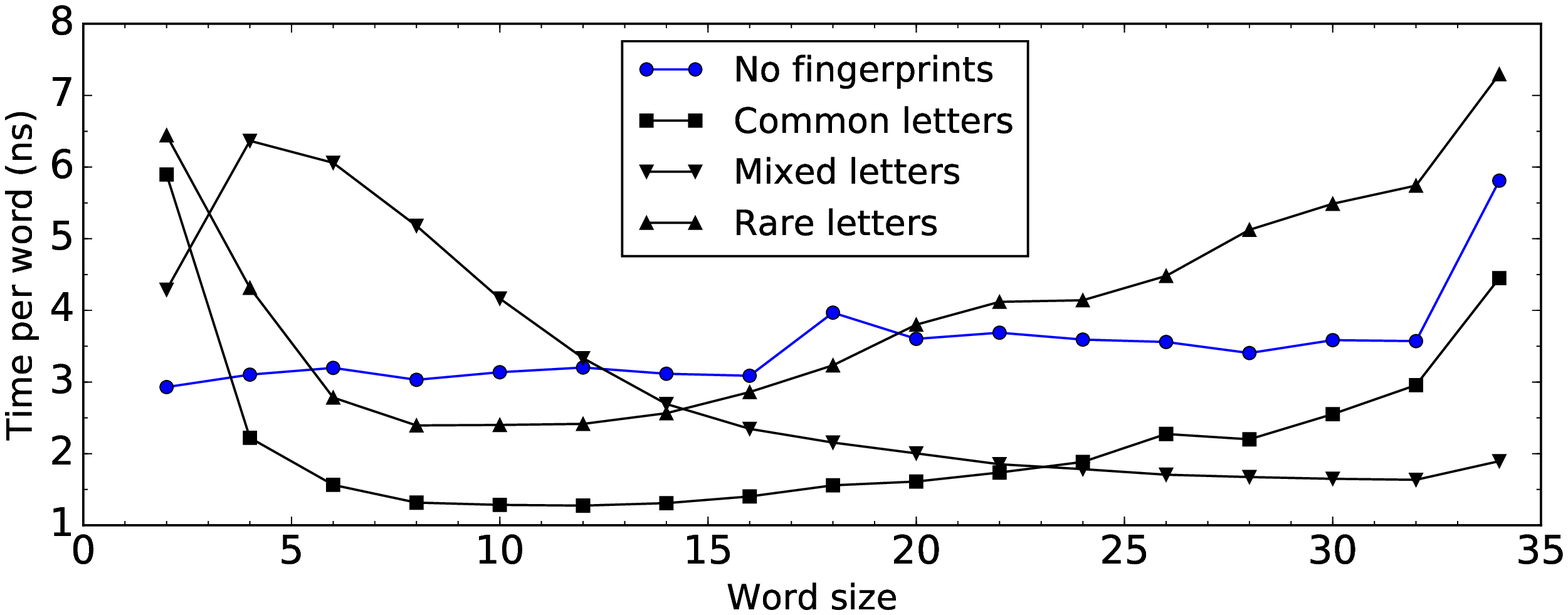}
    \includegraphics[scale=0.43]{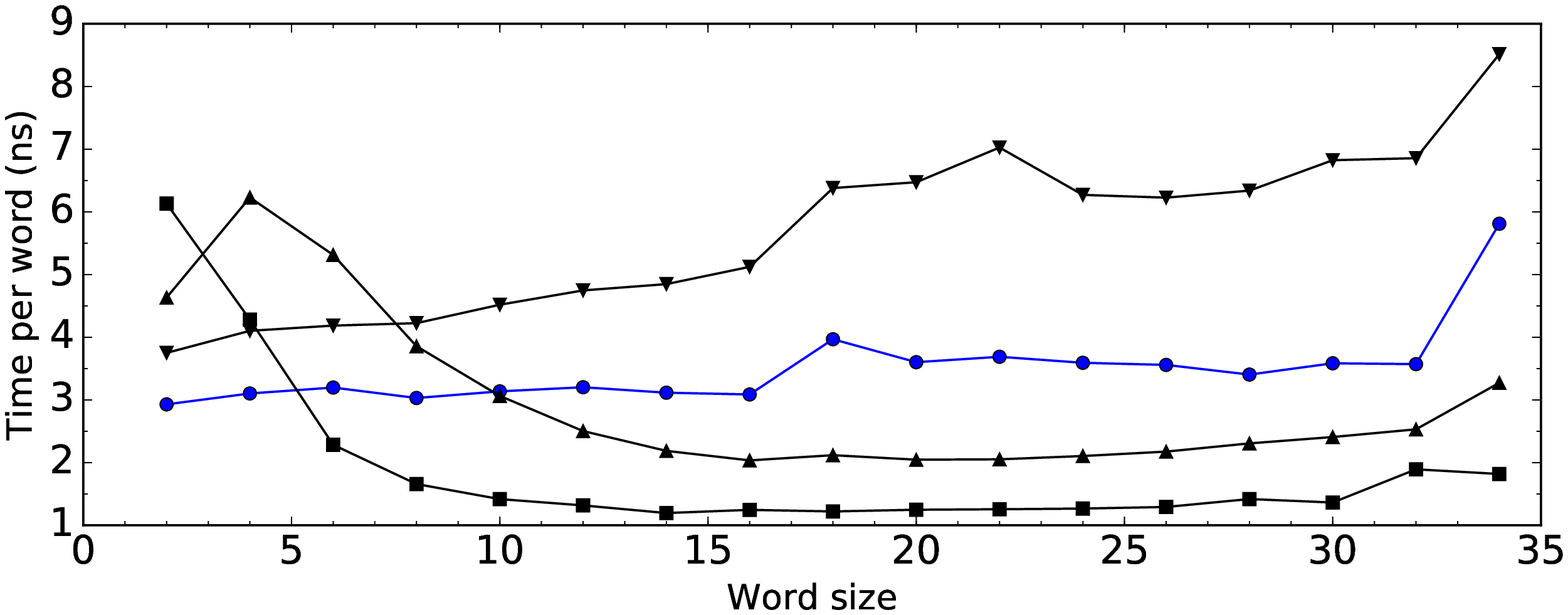}
    \includegraphics[scale=0.43]{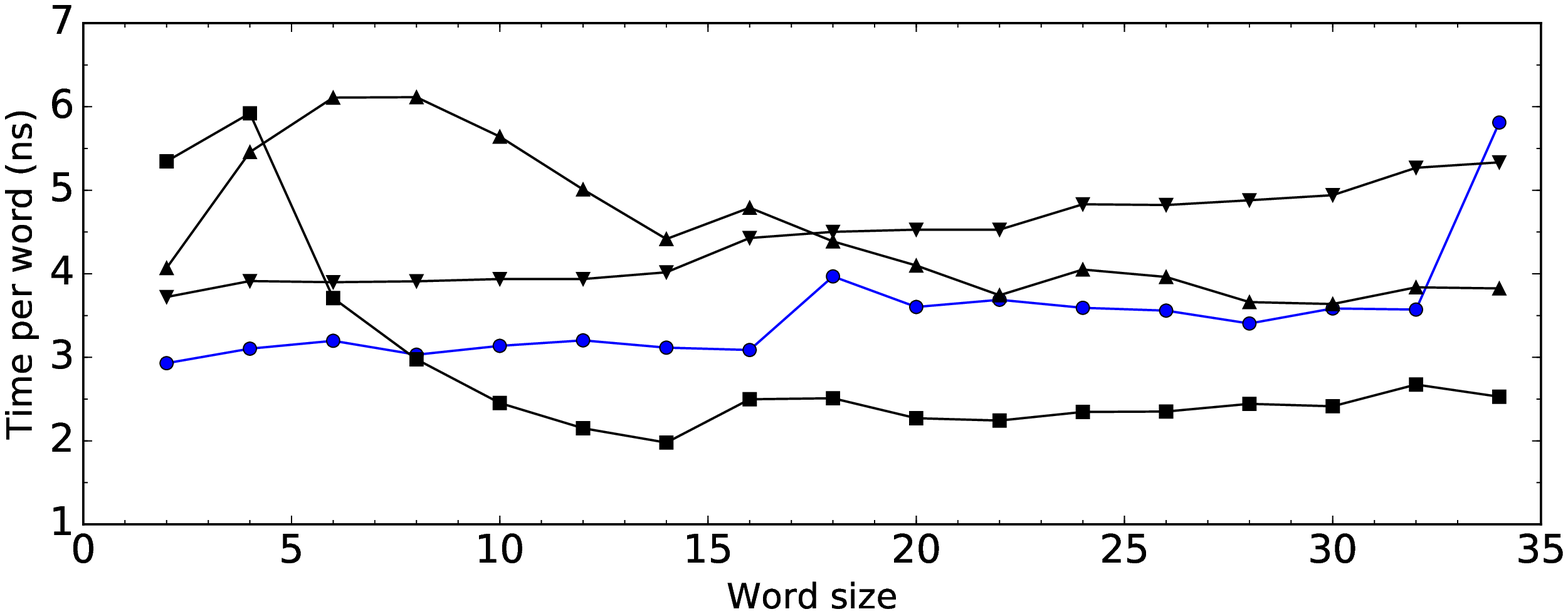}

    \caption{Comparison time vs word size for 1 mismatch (\textbf{Hamming} distance) for \textbf{synthetic} data. Words were generated over the English alphabet. Time refers to average comparison time between a single pair of words. The upper figure shows results for occurrence fingerprints, the middle figure shows results for count fingerprints, and the bottom figure shows results for position fingerprints.}
    \label{Fig:word_size_synth}
\end{figure}

In Tables~\ref{Tab:english_url_speedup}~and~\ref{Tab:english_url_speedup_lev} we present the speedup for $k = 1$ that was achieved for the English dictionary and the URL data.
Word lengths of $9$ ($90\,501$ words in total) and $69$ ($34\,044$ words in total) were used, respectively, which corresponded to mode length values in the tested dictionaries.
The speedup $S$ was calculated using the following formula: $S = T_n / T_f$, where $T_n$ refers to the average time required for a naive comparison, and $T_f$ refers to the average time required for comparison using fingerprints (e.g.,~$2.0$ means that the time required for comparison decreased twofold when fingerprints were used).
Higher speedup in the case of URL data was caused by a higher level of similarity between the data.
In particular, the data set comprised some URLs which referred to different resources that were located on the same server.
This resulted in certain words sharing a common prefix, requiring a naive algorithm to proceed with checking at least several first characters of each word.
Presented results also demonstrate a limitation of our technique, which is apparent in the case of shorter words, where using fingerprints may increase the comparison time.
Position fingerprints are not listed for the URL data, since they were completely ineffective due to multiple common prefixes between words and a high word size (almost no words were rejected).
Let us also note that Hamming distance results for the English words are consistent with those reported for synthetic English data.

\begin{table}[ht!]
	\centering
	\begin{tabular}{c|ccc}

	English & Common & Mixed & Rare \\
	\hline
	Occurrence & $1.20$ & $0.78$ & $0.51$ \\
	Occurrence halved & $1.16$ & $0.67$ & $0.56$\\
	Count & $0.86$ & $0.43$ & $0.60$\\
	Position & $0.75$ & $0.46$ & $0.67$\\
	
	\end{tabular}
	\vspace{5mm}
	
	\begin{tabular}{c|ccc}

	URLs & Common & Mixed & Rare \\
	\hline
	Occurrence & $1.49$ & $1.25$ & $2.05$\\
	Occurrence halved & $1.70$ & $1.79$ & $1.51$\\
	Count & $1.66$ & $1.65$ & $1.36$\\
	
	\end{tabular}
	\vspace{5mm}

    \caption{Speedup for various fingerprint types relative to a naive comparison for $k = 1$ using \textbf{Hamming} distance for \textbf{real-world} data. Values smaller than $1.0$ indicate that there was no speedup and the time required for comparison increased. The results in the upper table were calculated for the set of English language words of length $9$, and the results in the lower table were calculated for the set of URLs of length $69$.}
    \label{Tab:english_url_speedup}
\end{table}

In Table~\ref{Tab:english_url_rejected} we list percentages of words that were rejected for the same data sets for $k = 1$ as a hardware-independent method of comparing our approaches.
The rejection rate is naturally positively correlated with the speedup in comparison time.
Table~\ref{Tab:english_url_construction} presents the construction speed.
Construction time included the creation of all lookup tables which accompany the fingerprints in order to speed up their comparisons, as well as storing the fingerprints in a dynamic container.

\begin{table}[ht!]
	\centering
	\begin{tabular}{p{5em} | ccc}

	\centering English & Common & Mixed & Rare \\
	\hline
	\centering Occurrence & $4.78$ & $2.04$ & $0.95$ \\
	\centering Count & $2.39$ & $0.82$ & $0.25$\\
	
	\end{tabular}
	\vspace{5mm}
	
	\begin{tabular}{p{5.3em} | ccc}

	\centering URLs & Common & Mixed & Rare \\
	\hline
	\centering Occurrence & $3.61$ & $2.31$ & $10.38$\\
	\centering Count & $5.36$ & $5.42$ & $\hphantom{1}3.18$\\
	
	\end{tabular}
	\vspace{5mm}

    \caption{Speedup for various fingerprint types relative to a naive comparison for $k = 1$ using \textbf{Levenshtein} distance for \textbf{real-world} data. Values smaller than $1.0$ indicate that there was no speedup and the time required for comparison increased. The results in the upper table were calculated for the set of English language words of length $9$, and the results in the lower table were calculated for the set of URLs of length $69$.}
    \label{Tab:english_url_speedup_lev}
\end{table}
\setlength{\tabcolsep}{6pt}

\begin{table}[ht!]
	\centering
	\begin{tabular}{c|ccc}

	English & Common & Mixed & Rare \\
	\hline
	Occurrence (Ham, Lev) & $98.41\,\%$ & $91.91\,\%$ & $78.72\,\%$\\
	Occurrence halved (Ham) & $97.84\,\%$ & $87.90\,\%$ & $\hphantom{7}9.00\,\%$\\
	Count (Ham, Lev) & $93.60\,\%$ & $75.31\,\%$ & $\hphantom{7}5.94\,\%$\\
	Position (Ham) & $90.72\,\%$ & $54.93\,\%$ & $\hphantom{7}0.36\,\%$\\
	
	\end{tabular}
	\vspace{5mm}
	
	\begin{tabular}{c|ccc}

	URLs & Common & Mixed & Rare \\
	\hline
	Occurrence (Ham, Lev) & $70.79\,\%$ & $54.75\,\%$ & $89.31\,\%$\\
    Occurrence halved (Ham) & $79.73\,\%$ & $83.99\,\%$ & $72.71\,\%$\\
 	Count (Ham, Lev) & $80.34\,\%$ & $80.29\,\%$ & $66.76\,\%$\\
	
	\end{tabular}
	\vspace{5mm}

    \caption{Percentage of \textbf{rejected words} for various fingerprint types for $k = 1$ for \textbf{real-world} data. Rejection means that the true error was determined to be more than $k$ based only on a fingerprint comparison. Results in the upper table were calculated for the set of English language words of length $9$, and results in the lower table were calculated for the set of URLs of length $69$.}
    \label{Tab:english_url_rejected}
\end{table}

In general, the choice of the optimal strategy, viz.~fingerprint type, letters data set, and how many bits are used per count or position in a fingerprint, depends chiefly on the input data.
Larger fingerprints would allow for obtaining a better rejection rate, but this would come at a cost of increased space usage.
Once the rejection rate is close to the optimal $100\,\%$, larger fingerprints would provide only a negligible reduction in processing time.

In our case, the simplest approach, that is occurrence fingerprints with common letters, seemed to offer the best performance.
Still, we would like to point out that a practical evaluation on a specific data set would be advised in a real-world scenario.

\setlength{\tabcolsep}{10pt}
\begin{table}
	\centering

	\begin{tabular}{c|ccc}

	English & Common & Mixed & Rare \\
	\hline
	Occurrence & $274.36$ & $195.13$ & $214.33$\\
    Occurrence halved & $192.34$ & $193.32$ & $301.14$\\
 	Count & $138.43$ & $155.09$ & $291.35$\\
	Position & $172.00$ & $205.21$ & $256.18$\\	
	
	\end{tabular}
	\vspace{5mm}
	
	\begin{tabular}{c|ccc}

	URLs & Common & Mixed & Rare \\
	\hline
	Occurrence & $393.86$ & $380.34$ & $388.07$\\
    Occurrence halved & $379.77$ & $374.92$ & $397.82$\\
 	Count & $324.05$ & $356.17$ & $415.77$\\
	
	\end{tabular}
	\vspace{5mm}

    \caption{\textbf{Construction speed} given in MB/s for various fingerprint types for \textbf{real-world} data. Results in the upper table were calculated for the set of English language words of length $9$, and results in the lower table were calculated for the set of URLs of length $69$.}
\label{Tab:english_url_construction}
\end{table}

\section{Conclusions}
\label{Sec:future}

We have evaluated fingerprints in the context of dictionary matching.
Still, we would like to emphasize the fact that fingerprints are not a data structure in itself, rather, they are a string augmentation technique which we believe may prove useful in various applications.
For instance, they can be used in any data structure which performs multiple internal approximate string comparisons, providing substantial speedups at a modest increase in the occupied space.
In particular, for longer strings such as URL sequences the space overhead can be considered negligible.

Fingerprints take advantage of the letter distribution, and for this reason they were not effective for strings sampled over the alphabet with a uniformly random distribution.
They are also not recommended for the DNA data due to the small size of the alphabet and a long average word size.
These two combined properties result in a scenario where each word almost surely contains multiple occurrences of each possible letter.

In the future, we would like to extend the notion of a fingerprint by encoding information regarding not only single symbol distributions, but rather $q$-gram distributions.
The set of $q$-grams could be determined either heuristically or using an exhaustive search, and their use might provide speedup for any data set (possibly including DNA sequences).
We believe that it may be also beneficial for processing larger $k$ values.
Another possibility lies in combining different fingerprint types for a single word in order to further decrease comparison time at the cost of increased space usage.

\bibliographystyle{abbrv}
\bibliography{biblio}

\end{document}